\documentclass[11pt,a4paper]{article}
\usepackage[margin=1in]{geometry}
\usepackage[T1]{fontenc}
\usepackage{lmodern}
\usepackage{titlesec}
\titleformat{\section}{\large\bfseries\raggedright}{\thesection}{1em}{}
\titleformat{\subsection}{\normalsize\bfseries\raggedright}{\thesubsection}{1em}{}
\usepackage{amssymb,amsthm,booktabs}
\usepackage[authoryear,round]{natbib}
\newcounter{Appendix}
\renewcommand{\theAppendix}{\Roman{Appendix}}
\newenvironment{Appendix}{\refstepcounter{Appendix}\section*{Appendix \theAppendix}}{}

\usepackage{mathtools}
\usepackage{needspace}
\usepackage{placeins}
\usepackage{tikz}
\usetikzlibrary{arrows.meta,positioning,calc}
\usepackage{xcolor}
\usepackage{enumitem}
\usepackage{array}
\usepackage{tabularx}
\usepackage{microtype}
\usepackage{url}
\usepackage{graphicx}
\usepackage{hyperref}

\graphicspath{{figures/}}

\hypersetup{colorlinks=true,linkcolor=blue!50!black,citecolor=blue!50!black,
            urlcolor=blue!50!black}

\definecolor{layerG}{HTML}{0072B2}
\definecolor{layerP}{HTML}{E69F00}
\definecolor{layerF}{HTML}{009E73}
\definecolor{layerA}{HTML}{CC79A7}

\newcommand{\Gen}{\mathsf{G}}
\newcommand{\Pred}{\mathsf{P}}
\newcommand{\Rep}{\mathsf{F}}
\newcommand{\Act}{\mathsf{A}}
\newcommand{\cC}{\mathcal{C}}
\newcommand{\cM}{\mathcal{M}}

\newcommand{\cP}{\mathcal{P}}

\newcommand{\cX}{\mathcal{X}}
\newcommand{\bA}{\boldsymbol{A}}
\newcommand{\bB}{\boldsymbol{B}}

\newcommand{\bK}{\boldsymbol{K}}
\providecommand{\bM}{}
\renewcommand{\bM}{\boldsymbol{M}}
\newcommand{\bS}{\boldsymbol{S}}
\newcommand{\bW}{\boldsymbol{W}}
\newcommand{\bSigma}{\boldsymbol{\Sigma}}
\newcommand{\bb}{\boldsymbol{b}}
\newcommand{\bd}{\boldsymbol{d}}
\newcommand{\ba}{\boldsymbol{a}}
\newcommand{\bmu}{\boldsymbol{\mu}}
\newcommand{\bnu}{\boldsymbol{\nu}}
\newcommand{\bsigma}{\boldsymbol{\sigma}}
\newcommand{\blambda}{\boldsymbol{\lambda}}
\newcommand{\bU}{\boldsymbol{U}}
\newcommand{\bu}{\boldsymbol{u}}
\newcommand{\bone}{\boldsymbol{1}}
\providecommand{\bY}{}
\renewcommand{\bY}{\boldsymbol{Y}}
\providecommand{\by}{}
\renewcommand{\by}{\boldsymbol{y}}
\newcommand{\Gfield}{\Gamma}
\newcommand{\tgt}{\tau}
\newcommand{\R}{\mathbb{R}}
\providecommand{\E}{}
\renewcommand{\E}{\mathbb{E}}
\newcommand{\Haus}{\mathcal{H}}
\newcommand{\KL}{\operatorname{KL}}
\newcommand{\tr}{\operatorname{tr}}
\newcommand{\ran}{\operatorname{ran}}
\newcommand{\rank}{\operatorname{rank}}
\providecommand{\cov}{}
\renewcommand{\cov}{\operatorname{cov}}
\newcommand{\diag}{\operatorname{diag}}
\newcommand{\Rec}{\mathsf{R}}
\newcommand{\summ}{\mathsf{s}}
\newcommand{\pf}{\#}
\newcommand{\T}{^{\top}}

\newcolumntype{Y}{>{\raggedright\arraybackslash}X}
\newcolumntype{P}[1]{>{\raggedright\arraybackslash}p{#1}}

\theoremstyle{plain}
\newtheorem{theorem}{Theorem}
\newtheorem{proposition}[theorem]{Proposition}
\newtheorem{corollary}[theorem]{Corollary}

\theoremstyle{definition}

\newtheorem{remark}[theorem]{Remark}

\begin{document}

\title{Coherence is not enough: Aggregation constraints across predictive distributions, forecasts and decisions}

\author{Shun Hu\thanks{Email: \texttt{shunhu@buaa.edu.cn}.}
\and Yanfei Kang\thanks{Corresponding author. Email: \texttt{yanfeikang@buaa.edu.cn}.}}
\date{\small School of Economics and Management, Beihang University\\
Beijing 100191, China}
\maketitle

\begin{abstract}
\noindent

Forecasts made at different levels of aggregation are often required to agree—for example, regional forecasts should sum to the national total. Forecast reconciliation imposes such relationships, but the meaning of agreement depends on whether the constraint is applied to a predictive distribution, a reported summary such as a mean or quantile, or a decision based on the forecast. We show that these operations are not interchangeable. Even when every possible outcome from a predictive distribution satisfies an aggregation rule, its separately reported marginal quantiles need not add up; under nonlinear relationships, even the coordinatewise means may violate the constraint. We characterise when linear reconciliation preserves marginal quantiles and show that nonmedian quantiles generally cannot be preserved in a genuine hierarchy. We also explain why conditioning on an exact nonlinear constraint requires specifying how that constraint is observed or approximated, and identify when the predictive mean contains enough information for a constrained decision. In an application to Australian tourism data with 28 rolling forecast origins, the median discrepancy between two natural constructions of the national 95th percentile is 159.9\% of the reconciled 90\% prediction-interval width under bottom-up reconciliation. Rankings of mean- and quantile-based forecasts also reverse when squared-error loss is replaced by asymmetric inventory loss. These results provide a framework for deciding what should be made coherent and how competing methods should be evaluated.
\end{abstract}


\section{Introduction}\label{sec:intro}

Forecast reconciliation is usually posed as a geometric problem. For a
hierarchical or grouped collection the observations satisfy $\bY=\bS\bB$, the
columns of $\bS$ span a coherent subspace $\cC$, base forecasts need not lie
in $\cC$, and a reconciliation method maps them back
\citep{hyndman2011,wickramasuriya2019,panagiotelis2021}. Ordinary least
squares, weighted least squares and minimum trace (MinT) reconciliation
differ in the criterion that selects the coherent point.

Probabilistic reconciliation now includes mappings of samples or
distributions \citep{bentaieb2017,jeon2019,panagiotelis2023}, Bayesian updating
and conditioning \citep{corani2021,zambon2024efficient}, and bottom-level joint
models that generate coherent distributions by construction
\citep{bertani2025}. At the report layer, \citet{han2021} regularise quantile
forecasts toward aggregation consistency. \citet{honguyen2026} instead retain
a multivariate belief distribution and select its reconciliation map for
pinball performance; marginal quantiles extracted from the resulting coherent
distribution need not themselves add. The same algebraic relation therefore
enters different forecasting objects and objectives across this literature.

Several papers isolate parts of that distinction. \citet{kolassa2023} shows
that loss-optimal quantiles generally do not add and argues for joint
distributional forecasts. \citet{west2024} distinguishes conditioning,
externally imposed aggregates and constrained actions, while
\citet{athanasopouloskourentzes2023} emphasise evaluation and decision
objectives. \citet{zhang2024} contrast conditioning with mapping, and
\citet{nguyen2026information} separate the effect of aggregation constraints
from the information combined by reconciliation. We organise these
distinctions through a common question: where does the aggregation relation
come from, and which output must satisfy it?

A structural identity belongs to the data-generating mechanism, whereas an
additive table is a reporting requirement. An externally observed total
updates a predictive distribution, and a resource constraint restricts a
decision. Coherence verifies that the final output satisfies the relation
without identifying which operation produced it. We study three resulting
problems across the forecasting sequence: identifying a constrained
distribution, transporting an aggregation relation through summary extraction,
and selecting a loss-aligned action. The classical
affine--Gaussian--quadratic benchmark can hide these distinctions because
several operations share a common point. Nonmedian summaries and actions that
use information beyond the mean reveal which forecasting object was
constrained.

\subsection{Six coherent answers to one aggregate}

Consider two stores with independent future demands
\[
Y_1\sim\mathrm{Gamma}(4,10),\qquad
Y_2\sim\mathrm{Gamma}(36,5/3),
\]
where the second argument is scale, and let $P_0$ denote their joint
distribution under independence.
The means are $(40,60)\T$, the variances are $(400,100)\T$, and management
supplies the total $\tgt=85$.
Suppose a pair is mechanically used as an inventory allocation under
\begin{equation}\label{eq:invloss}
L(\by,\ba)=\sum_{i=1}^{2}\bigl\{c^{o}_i(a_i-y_i)_{+}
  +c^{u}_i(y_i-a_i)_{+}\bigr\},
\qquad
(c^{u}_1,c^{o}_1)=(9,1),\quad(c^{u}_2,c^{o}_2)=(1.5,1).
\end{equation}
The unconstrained optimal quantile levels are $0.9$ and $0.6$.

Table~\ref{tab:six} displays six interpretations of ``reconcile to
$85$''. Every output satisfies the same equation, yet the value assigned
to Store~1 ranges from $27.51$ to $51.11$. If every pair is used as an inventory
allocation and evaluated under the common loss in the last column, regret
ranges from zero to $61.3\%$. Thus a coherence check accepts six outputs
with materially different decision consequences.

\begin{table}[!ht]
\centering
\footnotesize
\caption{Six outputs satisfying the same two-store total $\tgt=85$ under
the same base distribution. The criterion in the fifth column is specific
to each operation, so its values are not comparable across rows. The last
column reports expected inventory loss \eqref{eq:invloss} when each pair is
used as an allocation under the unchanged base distribution; parentheses
give regret relative to the constrained Bayes action. A post-policy
evaluation of the final row would require a causal model.}
\label{tab:six}
\begin{tabularx}{\textwidth}{@{}P{1.9cm}YrrP{2.45cm}r@{}}
\toprule
Role of total & Operation & Store 1 & Store 2 & Method-specific criterion
  & Inventory loss \\
\midrule
Conditioning event & Conditional mean given $Y_1+Y_2=\tgt$
  & 27.51 & 57.49 & total log score $4.021$ & 146.98 $(61.3\%)$\\
Mean constraint & Minimum-$\KL$ distribution with $\E_q(Y_1+Y_2)=\tgt$
  & 28.70 & 56.30 & $\KL(q\Vert P_0)=0.269$ & 140.62 $(54.3\%)$\\
Structural identity & Euclidean projection of the base mean
  & 32.50 & 52.50 & squared displacement $112.5$ & 123.15 $(35.2\%)$\\
Structural identity & Variance-weighted projection of the base mean
  & 28.00 & 57.00 & weighted displacement $0.45$ & 144.31 $(58.4\%)$\\
Feasibility constraint & Constrained Bayes action under \eqref{eq:invloss}
  & 51.11 & 33.89 & expected loss $91.11$ & 91.11 $(0.0\%)$\\
Declared policy & Declared $50\!:\!50$ allocation policy
  & 42.50 & 42.50 & declared split & 96.87 $(6.3\%)$\\
\bottomrule
\end{tabularx}
\end{table}

The method-specific criteria prevent the fifth column from being read as a
common ranking: conditioning, minimum-$\KL$ revision, point projection,
Bayes action and policy declaration solve different problems. The final
column instead quantifies what happens if each output is used for the same
inventory decision. Appendix~\ref{app:gamma} gives the one-dimensional
calculations.

\subsection{Contributions and relationship to prior results}

The framework connects conditioning, quantile aggregation, decision evaluation
and mean elicitation. Zero-probability conditioning and the co-area formula
provide the distributional foundation
\citep{kolmogorov1956,changpollard1997,diaconis2013}, with Bayesian updating
and conditioning adapted to forecast reconciliation by
\citet{corani2021,zambon2024efficient}. Quantile nonadditivity and its link to
the elicited functional are discussed by \citet{kolassa2023} and
\citet{athanasopoulos2024}; the all-probability-level characterisation follows
from dependence theory \citep{dhaene2002,cheung2010}. Decision and evaluation
objectives are analysed by \citet{athanasopouloskourentzes2023}, \citet{west2024}
and \citet{honguyen2026}, and mean-consistent losses admit the Bregman
representation \citep{banerjee2005,gneiting2011,abernethy2012}.

Our contribution is a common analysis of these distinctions and sharp
boundaries for transporting an aggregation relation between forecasting
outputs. Proposition~\ref{prop:tube} expresses the Borel--Kolmogorov conditioning
ambiguity through the residual representation, which selects density along
a nonlinear constraint. Theorem~\ref{thm:commute} exactly characterises
the affine maps that preserve marginal quantiles uniformly and yields the
genuine-hierarchy no-go result. Equation~\eqref{eq:gap} expresses the Gaussian
commutator as a reconciliation-specific diversification diagnostic.
Corollary~\ref{thm:bregman} applies the classical Bregman risk decomposition
to constrained decisions, and Proposition~\ref{prop:nogo} supplies explicit
equal-mean witnesses outside mean sufficiency. Proposition~\ref{prop:gauss}
then synthesises Gaussian and quadratic identities to explain why these
operations can appear interchangeable. The tourism study quantifies the
report-level order effect through observable forecast paths and illustrates a
target-dependent ranking reversal between two reported outputs.

The argument follows the forecasting sequence. Section~\ref{sec:specification}
uses the six answers in Table~\ref{tab:six} to define a reconciliation
specification, distinguish two construction orders and state a comparison
contract. Sections~\ref{sec:conditioning}--\ref{sec:decision} examine
constraints on predictive distributions, reported forecasts and constrained
actions. Section~\ref{sec:gauss} explains why the classical benchmark hides
their differences at the point-summary level, and Section~\ref{sec:empirical}
quantifies the quantile order effect and gives a target-dependent evaluation
illustration in Australian tourism data. The final section discusses scope,
comparison and implications.

\section{A framework for specifying reconciliation}
\label{sec:specification}

\subsection{Four stages and the specification triple}

The six answers differ because the same equation enters at different stages
and is completed by different selection rules. A complete specification must
therefore locate the aggregation relation within the forecasting sequence.
That sequence has four stages:
\begin{equation}\label{eq:layers}
\Gen \longrightarrow \Pred \longrightarrow \Rep \longrightarrow \Act,
\qquad \Act \rightsquigarrow \Gen^{\pi},
\end{equation}
where $\Gen$ is the generative mechanism, $\Pred$ the predictive
distribution, $\Rep$ the reported forecast and $\Act$ the action or policy; a
policy selected at $\Act$ replaces $\Gen$ by a post-policy mechanism
$\Gen^{\pi}$. Figure~\ref{fig:pipeline} separates this forward flow from the
stage at which an aggregation relation enters.

\begin{figure}[!ht]
  \centering
  \resizebox{\textwidth}{!}{\begin{tikzpicture}[
  >=Latex,
  node distance=10mm and 12mm,
  stage/.style={draw, rounded corners=1.2pt, minimum width=24mm,
    minimum height=9mm, align=center, line width=0.65pt},
  stageG/.style={stage, draw=layerG!75, fill=layerG!11},
  stageP/.style={stage, draw=layerP!80, fill=layerP!13},
  stageF/.style={stage, draw=layerF!75, fill=layerF!11},
  stageA/.style={stage, draw=layerA!75, fill=layerA!11},
  meaning/.style={font=\footnotesize, align=center},
  flow/.style={->, line width=0.55pt},
  attach/.style={->, densely dashed, line width=0.5pt}
]
  \node[stageG] (g) {$\Gen$\\generative mechanism};
  \node[stageP, right=of g] (p) {$\Pred$\\predictive distribution};
  \node[stageF, right=of p] (f) {$\Rep$\\reported forecast};
  \node[stageA, right=of f] (a) {$\Act$\\action or policy};

  \draw[flow] (g) -- (p);
  \draw[flow] (p) -- (f);
  \draw[flow] (f) -- (a);

  \node[meaning, above=13mm of g] (str) {structural relation\\$h(\bY)=0$ a.s.};
  \node[meaning, above=13mm of p] (evd) {observed information or\\belief restriction};
  \node[meaning, above=13mm of f] (rep) {reporting requirement\\$h(\widehat{\by})=0$};
  \node[meaning, above=13mm of a] (dec) {decision requirement\\$h(\ba)=0$};
  \node[meaning, below=13mm of a] (int)
    {policy target\\select $\pi$; implement $\Gen\mapsto\Gen^\pi$};

  \draw[attach, draw=layerG!75] (str) -- (g);
  \draw[attach, draw=layerP!85] (evd) -- (p);
  \draw[attach, draw=layerF!75] (rep) -- (f);
  \draw[attach, draw=layerA!75] (dec) -- (a);
  \draw[attach, draw=layerA!75] (int) -- (a);
  \draw[flow, draw=layerG!75, rounded corners=2pt]
    (int.west) -| node[pos=.58, below, font=\scriptsize] {implementation} (g.south);
\end{tikzpicture}}
  \caption{Four forecasting stages and the points at which the same
  aggregation relation can enter. Dashed arrows attach the relation to the
  output required to satisfy it; solid arrows show forecast production and
  policy implementation. The framework records both the source of the relation
  and the output on which it is enforced.}
  \label{fig:pipeline}
\end{figure}
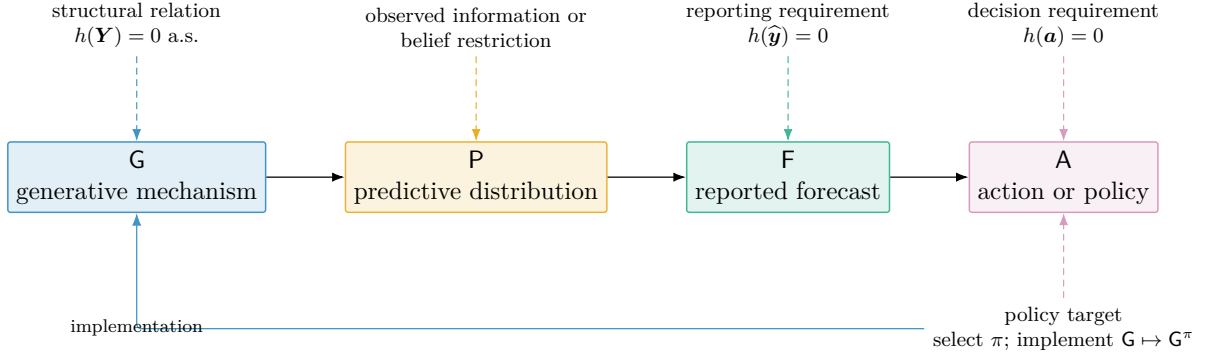

A reconciliation specification is the triple
\begin{equation}\label{eq:spec}
\text{specification}
= (\text{constraint source},\ \text{constrained output},
\ \text{selection criterion}).
\end{equation}
MinT applies a structural aggregation relation to a reported vector and
selects that vector by a covariance-weighted quadratic criterion.
Conditioning on an externally supplied total instead updates a predictive
distribution. A constrained Bayes rule imposes a feasibility constraint on
an action and selects it by expected operational loss.

The triple locates conditioning versus mapping \citep{zhang2024} and forecast
construction versus evaluation and decision objectives
\citep{athanasopouloskourentzes2023}. Its source coordinate also keeps the
aggregation relation distinct from information combined by the selected method
\citep{nguyen2026information}. Joint bottom-up modelling is one instance:
it places a structural aggregation relation at the predictive-distribution
output and selects the law through a joint model of the bottom-level series
\citep{bertani2025}.

When the constraint originates at the same stage as the required output, the
operation is an update within that stage. When the stages differ, the order of
enforcement and the output-forming operation must also be specified.

\subsection{Two possible orders of operation}

Write $\Rec$ for an operation that enforces the relation and $\summ$ for an
output-forming operation---extracting a summary, forming a marginal or
selecting an action. Figure~\ref{fig:square} compares the two possible orders.

\begin{figure}[!ht]
\centering
\begin{tikzpicture}[
  node distance=1.35cm and 3.2cm,
  every node/.style={font=\small},
  arr/.style={-{Latex[length=2.2mm]},thick}
]
\node (p) {$P$};
\node[right=of p] (rp) {$\Rec(P)$};
\node[below=of p] (sp) {$\summ(P)$};
\node[right=of sp] (out) {$\cX_{\mathrm{out}}$};
\draw[arr] (p) -- node[above] {$\Rec$} (rp);
\draw[arr] (p) -- node[left] {$\summ$} (sp);
\draw[arr] (rp) -- node[right] {$\summ$} (out);
\draw[arr] (sp) -- node[below] {$\Rec$} (out);
\end{tikzpicture}
\caption{Two possible orders. The upper path enforces the aggregation
relation before forming the downstream output, $\summ\!\circ\!\Rec$; the lower
path forms the output first and then enforces the relation,
$\Rec\!\circ\!\summ$. Both paths end in the common output space
$\cX_{\mathrm{out}}$; they commute only when their endpoints are equal.}
\label{fig:square}
\end{figure}
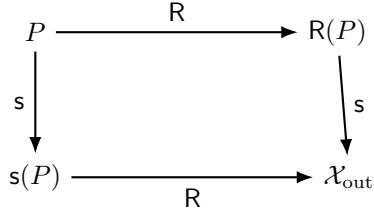

When the two paths agree, the order of operations is immaterial. When they
differ, both endpoints may still satisfy the aggregation relation, so a
coherence diagnostic alone cannot reveal the difference. The square formalises
the question whenever both routes are well defined. The framework also covers
questions that are not literal commutation problems, including whether a
constraint identifies a distribution and whether a report suffices for a
decision.

\subsection{Comparison contract}

The square distinguishes construction paths; comparing their outputs also
requires a common evaluation target. An evaluation target---forecast accuracy,
proper-score performance, calibration or decision cost---is declared
separately. It judges the output without determining how that output was
constructed. A direct comparison therefore requires agreement on
\begin{equation}\label{eq:comparison-contract}
(\text{source of constraint},\ \text{constrained output},\ \text{evaluation target}).
\end{equation}
Together, the specification triple and comparison contract form a diagnostic
audit: the former records how each output is constructed, while the latter
determines whether two outputs address the same task.

Table~\ref{tab:six} makes the distinction operational. All six pairs satisfy
$Y_1+Y_2=\tgt$, so a coherence check cannot separate them. Their
method-specific criteria also have different units and cannot form a common
ranking. A common evaluation target instead exposes the cost of substituting
one output for another. Within point reconciliation, changing the weighting
matrix from $\bW=I$ to $\bW=\diag(400,100)$ changes the Store~1 allocation
from $32.50$ to $28.00$ and changes inventory loss by $21.17$. Replacing the
constrained decision by the conditional mean changes the Store~1 allocation
by $23.6$ and changes
inventory loss by $55.9$, a regret of $61.3\%$.

The degrees of freedom make this underdetermination explicit.
With $n$ series and $m$ bottom-level series, aggregation coherence is
$\bA\bY=0$ for a full-row-rank $\bA\in\R^{(n-m)\times n}$ with
$\cC=\ker\bA$, so coherence removes $n-m$ directions and leaves $m$.
Because $\bB$ coordinatises
$\cC$, the method---not the constraint---determines every bottom-level
value. This is the elementary content of $\dim\ker\bA=m$ and is implicit in
the geometric account of \citet{panagiotelis2021}. The count explains why a
coherent set admits many answers, but it does not identify the probability
distribution carried by that set. The information-combination effect identified
by \citet{nguyen2026information} is separate from this dimension count. Point
projection, distributional revision and conditioning act in different spaces,
so each requires its own selection or observation rule. The first question is
therefore how a constraint determines a predictive distribution.

\FloatBarrier

\section{A constraint set does not determine a predictive distribution}
\label{sec:conditioning}

Bayesian updating and conditioning specify predictive reconciliation laws
\citep{corani2021,zambon2024efficient}, while \citet{zhang2024} contrast
conditioning with a mapping construction for discrete forecasts. Within a
conditioning construction, we study representation invariance: when a
continuous nonlinear constraint is represented in different but equivalent
ways, do those representations induce the same law? Write the coherent set as
$\cM=\{\by: h(\by)=0\}$. The Borel--Kolmogorov problem arises because
conditioning on a zero-probability set requires a specification of how the
set is observed
\citep{kolmogorov1956,changpollard1997,diaconis2013}. Shrinking-tolerance
conditioning makes this specification explicit through the residual $h$.

\begin{proposition}[Shrinking-tolerance conditioning and density ambiguity]
\label{prop:tube}
Let $\bY$ have a continuous density $p_0$ on $\R^n$, let
$h:\R^n\to\R^k$ be continuously differentiable with rank $k$ on the
embedded manifold $\cM=h^{-1}(0)$, and set
$J_h(y)=\det\bigl(Dh(y)Dh(y)\T\bigr)^{1/2}$.
Assume
\begin{enumerate}[label=(A\arabic*),leftmargin=2.2em,itemsep=1pt]
\item $\int_{\cM}p_0J_h^{-1}\,d\Haus^{n-k}\in(0,\infty)$;
\item for each bounded continuous $f$, the map
$z\mapsto\int_{h^{-1}(z)}fp_0J_h^{-1}\,d\Haus^{n-k}$ is continuous at $0$
and dominated near $0$ by an integrable function.
\end{enumerate}
Then:
\begin{enumerate}[label=(\roman*),leftmargin=2.2em,itemsep=2pt]
\item for every bounded continuous $f$,
\begin{equation}\label{eq:tube}
\lim_{\epsilon\downarrow0}\E\bigl(f(\bY)\mid \lVert h(\bY)\rVert\le\epsilon\bigr)
=\frac{\displaystyle\int_{\cM} f\,p_0\,J_h^{-1}\,d\Haus^{n-k}}
       {\displaystyle\int_{\cM} p_0\,J_h^{-1}\,d\Haus^{n-k}} .
\end{equation}
\item if $\cM$ is compact and the limiting law in part (i) has a positive
smooth density, every probability distribution on $\cM$ with a positive
smooth density ratio relative to that law is induced by a smooth positive
scalar rescaling $\widetilde h=\phi h$ of the same constraint set.
\end{enumerate}
\end{proposition}

The proof is given in Appendix~\ref{app:proofs}.

Write $p_0^{h}$ for the normalised density on $\cM$ in \eqref{eq:tube}.

\begin{remark}[The ambiguity is in $h$, not the tolerance shape]\label{rem:tubeshape}
Replacing the Euclidean ball by $\epsilon K$, where $K$ is a bounded
measurable neighbourhood of the origin with $|K|>0$, multiplies numerator
and denominator by $|\epsilon K|$ and leaves
\eqref{eq:tube} unchanged. The limit therefore does not depend on the
shape of the tolerance region; it depends on the function $h$ used to
express the constraint. If
the analyst instead posits a measurement model $h(\bY)+\sigma\varepsilon$
and lets $\sigma\downarrow0$, the same expression is obtained with $J_h$
replaced by the Jacobian of the measured residual. Declaring the observed
residual therefore completes the specification of the operation.
\end{remark}

Thus the coherent set fixes support but leaves the density along that set
undetermined: the constraint alone does not define the $(\Pred,\Pred)$
update.

\subsection{One rate identity, two conditional laws}\label{sec:ratios}

Consider forecasts of the number unemployed $U$, the labour force $N$ and
the unemployment rate $r$. These quantities satisfy $U=rN$, although
forecasts from separate models may not. Across regions, the counts add
and the national rate is $\sum_i U_i/\sum_i N_i$.
When future labour-force sizes are uncertain, these relations form a system
of linear and nonlinear constraints. Mortality forecasts have the same
structure through deaths, population exposure and mortality rates.
\citet{girolimetto2025} study both applications. Tourism volumes and regional
shares provide another example, examined by \citet{biswas2026}.

Using a continuous approximation for counts on $N>0$, consider the residuals
\[
h_1(U,N,r)=U-rN,\qquad h_2(U,N,r)=U/N-r.
\]
The first measures discrepancy in numbers of people and the second in rates.
Both vanish on the same surface, but $h_2=h_1/N$. Under the assumptions of
Proposition~\ref{prop:tube}, their shrinking-tolerance laws therefore satisfy
\[
p_0^{h_2}(U,N,r)\ \propto\ N\,p_0^{h_1}(U,N,r).
\]
The rate tolerance $|U/N-r|\leq\epsilon$ permits a count discrepancy
$|U-rN|\leq\epsilon N$. It therefore allows a wider band in count units
when the labour force is larger, producing the factor $N$ in the conditional
density. The identity fixes the surface, while the tolerance rule determines
the distribution along it. Appendix~\ref{app:revenue} gives a lognormal
product example with explicit conditional laws.

For the declared canonical affine residual $h(\by)=\bA\by-\bb$, the co-area
Jacobian is the constant $\det(\bA\bA\T)^{1/2}$ and introduces no
position-dependent reweighting along the constraint set. For nonlinear
relations, conditioning uses the chosen observation map \citep{biswas2026},
whereas projection uses a distance to the constraint surface
\citep{girolimetto2025,nespoli2026}.
Once an observation rule fixes the predictive distribution, the order problem
moves to the $\Pred\to\Rep$ transition: whether summary extraction commutes
with reconciliation.

\section{From a predictive distribution to a reported forecast}\label{sec:report}

Once specified, a predictive distribution is communicated through summaries.
At the $\Pred\to\Rep$ transition, the two paths in
Figure~\ref{fig:square} compare reconciliation before and after summary
extraction. Let $\bW$ be positive definite,
$\cC=\{\by:\bA\by=\bb\}$ with $\rank(\bA)=r$, and
\begin{equation}\label{eq:TW}
\begin{aligned}
T_{\bW}(\by)&=\by-\bW\bA\T(\bA\bW\bA\T)^{-1}(\bA\by-\bb)
  =\bM\by+\bd,\\
\bM&=I-\bW\bA\T(\bA\bW\bA\T)^{-1}\bA .
\end{aligned}
\end{equation}
Then $\bM^2=\bM$ and $\ran\bM=\ker\bA$, which is the linear subspace parallel to
$\cC$ and coincides with $\cC$ only when $\bb=0$.
For $\summ=\E$ and any affine map, the two paths in
Figure~\ref{fig:square} agree by linearity of expectation. Changing either
the summary or the geometry of the constraint can make the paths differ.
Table~\ref{tab:twobytwo} states the pattern.

\begin{table}[!ht]
\centering\small
\caption{Do reconciliation and summary extraction commute? Only an affine
constraint combined with the mean summary guarantees equality.}
\label{tab:twobytwo}
\begin{tabular}{@{}lll@{}}
\toprule
 & mean summary & marginal-quantile summary\\
\midrule
affine constraint & yes (linearity) & not uniformly in a genuine hierarchy
  (Thm.~\ref{thm:commute})\\
curved constraint & generally no (Prop.~\ref{prop:curv}) & generally no\\
\bottomrule
\end{tabular}
\end{table}

\subsection{Curved constraints and mean reports}\label{sec:curvature}

The equality for the mean in the upper-left cell of
Table~\ref{tab:twobytwo} relies on the constraint being affine. For the
unemployment relation in Section~\ref{sec:ratios},
\[
\E(U)=\E(rN)=\E(r)\E(N)+\cov(r,N).
\]
For example, two equally likely scenarios $(U,N,r)=(10,100,0.1)$ and
$(40,200,0.2)$ each satisfy the relation. Their coordinate means are
$(25,150,0.15)$, whereas $25/150=1/6$. The mean rate and the ratio of
mean counts are different summaries of the same coherent distribution.
The following result describes the local displacement of a mean from a
smooth constraint in terms of curvature and predictive dispersion.

\begin{proposition}[Curvature--variance displacement]\label{prop:curv}
Let $\cM\subset\R^n$ be a $C^2$ hypersurface, $\by_0\in\cM$, $\bnu$ a unit
normal, and write the local graph as
$\bY^{\epsilon}=\by_0+\bU^{\epsilon}+g(\bU^{\epsilon})\bnu$
with $\bU^{\epsilon}$ tangential, $g(0)=0$, $Dg(0)=0$. Suppose
$\E(\bU^{\epsilon})=0$,
$\bSigma_{T,\epsilon}=\cov(\bU^{\epsilon})$,
$\bU^{\epsilon}\to0$ in probability, and the Taylor remainder satisfies
$g(\bu)=\tfrac12\bu\T D^2g(0)\bu+R(\bu)$ and
$\E|R(\bU^{\epsilon})|=o[\E\lVert \bU^{\epsilon}\rVert^2]$. With
$\mathrm{II}_{\by_0}=D^2g(0)$ the signed second fundamental form,
\begin{equation}\label{eq:curv}
\bigl\langle \E(\bY^{\epsilon})-\by_0,\ \bnu\bigr\rangle
=\tfrac12\tr\bigl(\mathrm{II}_{\by_0}\bSigma_{T,\epsilon}\bigr)
 +o\bigl[\E\lVert \bU^{\epsilon}\rVert^2\bigr],
\end{equation}
which in principal coordinates is $\tfrac12\sum_i\kappa_i\sigma_i^2$.
\end{proposition}

The centring condition selects $\by_0$ as the tangential centre of the local
distribution and is essential. Affine constraints have vanishing second
fundamental form, which is why a coherent predictive distribution can be
summarised by its mean without a second reconciliation in a linear hierarchy.
Curvature therefore gives one source of report-level noncommutation. Even under
affine geometry, however, a nonlinear reported functional can create another;
marginal quantiles provide the leading case.

\subsection{Marginal quantiles}\label{sec:quantiles}

We next keep the reconciliation map affine and change the reported functional.
The nonadditivity of marginal quantiles and the resulting conflict between
additive point coherence and loss-optimal quantile reports are known
\citep{kolassa2023,athanasopoulos2024}.
The following result sharpens that observation by characterising exactly when
an affine map makes the square in Figure~\ref{fig:square} commute uniformly
over a distribution class.
Use the lower-quantile convention
\[
Q_X(\alpha)=\inf\{x:F_X(x)\ge\alpha\},\qquad
q_\alpha(P)=\bigl(Q_{Y_1}(\alpha),\dots,Q_{Y_n}(\alpha)\bigr)\T.
\]
For $\bY\sim P$, $T_{\pf}P$ denotes the distribution of $T(\bY)$.

\begin{theorem}[When affine maps preserve marginal quantiles]\label{thm:commute}
Fix $\alpha\neq1/2$ and let $T(\by)=\bM\by+\bd$. Suppose a class of base distributions
contains every nondegenerate Gaussian distribution on $\R^n$ with diagonal
covariance. Then
\begin{equation}\label{eq:commute}
q_\alpha(T_{\pf}P)=T\bigl(q_\alpha(P)\bigr)
\end{equation}
holds for every distribution in the class if and only if each row of $\bM$ is zero or
has exactly one nonzero entry, and that entry is positive. The row
condition is sufficient for arbitrary base distributions, including non-Gaussian
ones.
\end{theorem}

If $\bM$ is in addition idempotent, each nonzero row is $c_ie_{\pi(i)}\T$
with $c_i>0$ and row $\pi(i)$ of $\bM$ equals $e_{\pi(i)}\T$. A map
satisfying \eqref{eq:commute} may therefore copy and positively rescale anchor
coordinates; it need not be a coordinate projection. Equivalently, the two
orders can agree for a map that copies and rescales an anchor coordinate.
For example,
$\bigl(\begin{smallmatrix}1&0\\1&0\end{smallmatrix}\bigr)$ is idempotent
and satisfies \eqref{eq:commute}, yet its range is the line $x_1=x_2$.

\begin{corollary}[Affine reconciliation cannot preserve all nonmedian marginal quantiles]\label{cor:mint}
Let $\cC$ be the coherent set of a hierarchy with $m\ge2$ freely varying
bottom-level series and a total coordinate equal to their sum. Let
$T(\by)=\bM\by+\bd$ be any affine map with
$\bM$ idempotent and $\ran\bM=\ker\bA$. For every $\alpha\neq1/2$, some
nondegenerate Gaussian base distribution satisfies
\[
q_\alpha(T_{\pf}P)\ne T\bigl(q_\alpha(P)\bigr).
\]
In particular this holds for
$T_{\bW}$ in \eqref{eq:TW} for every positive definite $\bW$, hence for OLS,
WLS and MinT.
\end{corollary}

\begin{proof}
Idempotency with $\ran\bM=\ker\bA$ makes $\bM$ the identity on $\ker\bA$. If
each row of $\bM$ had at most one nonzero entry, the aggregate coordinate would satisfy
$y_{\text{tot}}=c\,y_{\pi(\text{tot})}$ for all $\by\in\ker\bA$, contradicting
$y_{\text{tot}}=\sum_{j}y_j$ with $m\ge2$ free coordinates. Apply
Theorem~\ref{thm:commute}.
\end{proof}

For Gaussian distributions the two sides of \eqref{eq:commute} agree at the
median because each marginal median is its mean; the identity can still
fail at $\alpha=1/2$ for non-Gaussian distributions.

\subsection{The quantile aggregation gap is the diversification benefit}

Corollary~\ref{cor:mint} shows that the two orders cannot agree uniformly
over Gaussian base distributions for a genuine hierarchy. The next result
quantifies the gap for a given Gaussian distribution and identifies the
comonotonic endpoint characterised in the following subsection.

\begin{proposition}[Gaussian quantile gap and comonotonic additivity]\label{prop:gap}
Let $T(\by)=\bM\by+\bd$.
\begin{enumerate}[label=(\roman*),leftmargin=2.2em,itemsep=1pt]
\item If $P=N_n(\bmu,\bSigma)$ and
$\bsigma=(\sqrt{\bSigma_{11}},\ldots,\sqrt{\bSigma_{nn}})\T$, then
\begin{equation}\label{eq:gap}
T(q_\alpha(P))-q_\alpha(T_{\pf}P)
=z_\alpha\left\{\bM\bsigma-
\sqrt{\diag(\bM\bSigma\bM\T)}\right\},
\end{equation}
where the square root is componentwise.
\item If $\bM$ has nonnegative entries, coordinate $i$ of
\eqref{eq:gap} is
\[
z_\alpha\left\{\sum_j\bM_{ij}\sigma_j-
\left(\sum_{j,k}\bM_{ij}\bM_{ik}\bSigma_{jk}\right)^{1/2}\right\}.
\]
For $\alpha>1/2$ this quantity is nonnegative. It vanishes exactly when
the nonzero centred coordinates receiving positive weight are positive
scalar multiples of a common element of $L^2$; if $\bSigma$ is positive definite, this is equivalent
to row $i$ having at most one nonzero entry.
\item If $P$ is comonotonic and $\bM$ has nonnegative entries,
\eqref{eq:commute} holds exactly.
\end{enumerate}
\end{proposition}

Equation~\eqref{eq:gap} applies directly to correlated Gaussian forecasts
and reconciliation maps with signed weights. Under nonnegative aggregation,
part (ii) identifies the bracket as the difference between the sum of
component standard deviations and the standard deviation of the sum---the
diversification benefit of the aggregate. Positive definiteness makes this
gap strict for every genuine aggregation row. Reconciling published upper
marginal quantiles then overstates the aggregate quantile by $z_\alpha$ times
that diversification benefit. Part (iii) closes the gap at the comonotonic
extreme, where
nonnegative sums preserve the common rank ordering. The results bracket
the cases in which the two orders agree: agreement for every distribution
requires each row to have at most one positive nonzero entry, whereas a
comonotonic distribution permits any nonnegative map.

\subsection{Enforcing additive marginal quantiles}\label{sec:como}

The preceding results show that reconciliation and marginal-quantile
extraction generally do not commute. Requiring the published quantiles to
add forces the lower path in Figure~\ref{fig:square} to agree arithmetically
with the aggregation relation. We now ask when the true marginal quantiles
can satisfy this requirement at every probability level. Let
$Y_{+}=\sum_{i=1}^{m}Y_i$ and consider
\begin{equation}\label{eq:alllevel}
Q_{Y_{+}}(u)=\sum_{i=1}^{m}Q_{Y_i}(u),\qquad 0<u<1 .
\end{equation}
Probabilistic reconciliation can preserve general dependence through empirical
copulas and sample reordering, as in \citet{bentaieb2021}. There, each marginal
sample receives the ranks specified by the estimated copula. Pairing all
marginals at one common rank instead gives the comonotonic construction.
Ranked-sample reconciliation \citep{jeon2019} gives a finite empirical version
of the common-rank construction by sorting marginal samples before reconciling
each rank-matched vector. The next proposition states the corresponding
population implication when additivity is required at every probability level.

\begin{proposition}[Comonotonicity characterisation]\label{prop:como}
Let $Y_1,\dots,Y_m$ be integrable with continuous marginal distribution
functions, strictly increasing on their supports. Then
\[
\eqref{eq:alllevel}\ \text{holds for every }u\in(0,1)
\quad\Longleftrightarrow\quad
\exists\,U\sim\mathrm{Unif}(0,1):\
Y_i=Q_{Y_i}(U)\ \text{a.s. for every }i.
\]
\end{proposition}

Sufficiency is the standard quantile representation \citep{dhaene2002};
necessity follows because \eqref{eq:alllevel} equates the distribution of the sum
with that of the comonotonic sum, and \citet{cheung2010} characterises the
vector from the distribution of its sum. Thus all-level additivity of the
true quantiles characterises comonotonicity. Under this dependence structure,
quantile extraction commutes with every nonnegative linear map, as stated in
Proposition~\ref{prop:gap}(iii).

Under the same regularity conditions, a noncomonotonic predictive
distribution cannot have all of its true marginal quantile curves satisfy
\eqref{eq:alllevel}. Any additive set of reported quantile curves must
therefore differ from a true marginal quantile for at least one series and
probability level. Strict consistency of the pinball loss then makes its
expected marginal loss strictly larger at that series and probability level
\citep{gneitingraftery2007}. The resulting trade-off is between arithmetic
consistency of the published quantiles and proper marginal evaluation;
calling both properties ``coherence'' conceals it.

This distinction also separates two approaches to quantile-focused
reconciliation. The penalty in \citet{han2021} trades quantile accuracy against
aggregation inconsistency. \citet{honguyen2026} retain a coherent joint belief
distribution and select its reconciliation map under pinball loss; its true
marginal quantiles may remain nonadditive.

Additive quantiles are not uniformly conservative. If $Y_1,Y_2$ are
independent Bernoulli with success probability $0.04$, both $0.95$
quantiles are zero while the $0.95$ quantile of the sum is one; an
arbitrarily small continuous perturbation preserves the strict inequality.
This is the familiar failure of subadditivity behind value-at-risk not
being a coherent risk measure \citep{artzner1999}. Bounds on the quantile
of a sum under unspecified dependence \citep{embrechts2013} give the
achievable range and show that the comonotonic sum is not the worst case,
so the sign of a quantile-coherence adjustment is problem-dependent. For
the gamma marginals in Table~\ref{tab:six}, the sharp dependence bounds at
probability level $0.95$ are
\[
114.01\le Q_{Y_1+Y_2}(0.95)\le166.41.
\]
Independence gives $140.72$, whereas the comonotonic sum of marginal
quantiles is $154.88$. The latter lies $11.54$ below the attainable upper
bound, making the dependence effect quantitative in the running example.

Finally, \eqref{eq:alllevel} on a finite grid of probability levels does
not characterise a copula. It still alters marginal reports and should be
assessed at the probability levels on which it is imposed.

Together, the mean and quantile results complete the report-level analysis.
Report-level commutation does not settle the next transition: a coherent mean
report may still differ from the action selected under operational loss.

\section{From a predictive distribution to a constrained action}
\label{sec:decision}

Decision alignment depends on the point functional, evaluation criterion and
constrained objective declared for the task
\citep{kolassa2023,athanasopouloskourentzes2023,west2024,honguyen2026}. A
constrained action combines the predictive distribution with a loss and a
feasible set. This section characterises when the predictive mean suffices to
recover that action and when the rest of the distribution matters. Outside
mean sufficiency, using a reconciled point forecast as an allocation substitutes
a report for the Bayes decision.

Let the constrained Bayes action under $P$ be
\begin{equation}\label{eq:bayes}
\ba^{\star}(P)=\arg\min_{\bA\ba=\bb}\ \E_P\bigl(L(\bY,\ba)\bigr),
\end{equation}
assumed unique. Call a rule \emph{mean-based} if it has the form
$\Psi(\E_P(\bY))$ for a fixed $\Psi$ not depending on $P$;
$T_{\bW}(\E_P\bY)$ is
the leading example. A loss is strictly consistent for the mean when its
unique unconstrained Bayes action is the predictive mean. If two
distributions with the same mean have expected-risk functions that differ
only by an action-independent constant, every constrained Bayes problem can
be solved from the mean.

The Bregman risk decomposition of \citet{banerjee2005} yields the following
constrained decision rule and its first-order system.

\begin{corollary}[Constrained decisions under Bregman loss]
\label{thm:bregman}
Let $\mathcal D\subset\R^n$ be open and convex, let
$\varphi:\mathcal D\to\R$ be twice continuously differentiable and strictly convex,
and let $h:\mathcal D\to\R$ be measurable. For $\by,\ba\in\mathcal D$, define
\[
L(\by,\ba)=D_{\varphi}(\by,\ba)+h(\by)
=\varphi(\by)-\varphi(\ba)
-\langle\nabla\varphi(\ba),\by-\ba\rangle+h(\by).
\]
Let $P$ satisfy $P(\bY\in\mathcal D)=1$,
$\E_P\lVert\bY\rVert<\infty$, $\bmu=\E_P(\bY)\in\mathcal D$, and
$\E_P|\varphi(\bY)+h(\bY)|<\infty$. Let
$\bA\in\R^{r\times n}$ have full row rank and suppose
$\mathcal F=\{\ba\in\mathcal D:\bA\ba=\bb\}$ is nonempty. Then, for every $\ba\in\mathcal D$,
\[
\E_P L(\bY,\ba)=\E_P\{\varphi(\bY)+h(\bY)\}
-\varphi(\ba)-\langle\nabla\varphi(\ba),\bmu-\ba\rangle .
\]
Thus risk differences, and hence the set of minimisers over $\mathcal F$, depend on
$P$ only through $\bmu$. Every constrained Bayes action
$\ba^{\star}\in\mathcal F$ satisfies, for some $\blambda\in\R^r$,
\[
\nabla^2\varphi(\ba^{\star})(\bmu-\ba^{\star})
=\bA\T\blambda,
\qquad \bA\ba^{\star}=\bb.
\]
If $\mathcal D=\R^n$ and
$\varphi(\by)=\tfrac12\by\T\bW^{-1}\by$ for positive definite $\bW$, the
constrained Bayes action is unique and equals
$\ba^{\star}(P)=T_{\bW}(\bmu)$.
\end{corollary}

Equivalently, the expected loss is $D_\varphi(\bmu,\ba)$ plus a term
independent of $\ba$. Thus the constrained problem minimises this divergence
from the mean over the feasible set. Classical representation results give
conditions under which losses strictly consistent for the mean have Bregman
form \citep{banerjee2005,gneiting2011,abernethy2012}.

Weighted reconciliation of the predictive mean is therefore the
constrained Bayes rule for a specific loss. The next result shows when a
mean-based rule cannot recover the constrained decision.

\begin{proposition}[No-go for mean-based reconciliation]\label{prop:nogo}
Let $\cP_0$ be a class of predictive distributions, and let
$P,Q\in\cP_0$ have the same finite mean and different unique constrained
Bayes actions. Then
\[
\begin{gathered}
\E_P(\bY)=\E_Q(\bY),\qquad
\ba^{\star}(P)\ne\ba^{\star}(Q)\\
\Longrightarrow\quad
\nexists\,\Psi\ \text{such that}\
\Psi(\E_R\bY)=\ba^{\star}(R)
\quad\text{for every }R\in\cP_0.
\end{gathered}
\]
In particular, $T_{\bW}(\E_P\bY)$ cannot recover the constrained Bayes
action throughout $\cP_0$ for any fixed positive definite $\bW$. For the
inventory loss \eqref{eq:invloss}, the premise holds within the independent
product family $Y_i\sim\mathrm{Gamma}(k_i,\mu_i/k_i)$ with fixed means.
\end{proposition}

A concrete witness fixes the means at $(40,60)\T$ and uses shape vectors
$(4,36)$ and $(36,4)$. Under $a_1+a_2=85$ and \eqref{eq:invloss}, the two
product distributions have unique Bayes actions $(51.114,33.886)\T$ and
$(45.912,39.088)\T$, respectively. Their marginal quantile functions, not
their common mean, determine the different allocations.

Two consequences follow. First, a non-quadratic loss can still yield
a constrained decision determined by the mean, since all Bregman losses
are strictly consistent for the mean
\citep{bregman1967,gneiting2011}; the relevant boundary is whether the
constrained action uses distributional information beyond the mean, as
under asymmetric absolute (quantile) and asymmetric squared (expectile) losses
\citep{neweypowell1987,gneiting2011}. Second, the familiar matrix algebra of
MinT and of a quadratic Bayes action coincides when the precision matrix and
loss curvature are deliberately matched, as in Proposition~\ref{prop:gauss}:
MinT treats high-forecast-variance directions as less costly to adjust,
whereas a quadratic Bayes rule penalises high-loss-curvature directions more
heavily.

Value-oriented reconciliation lets a downstream criterion enter earlier
without collapsing the distinction. \citet{wenpinson2026} select a coherent
combination function by Nash bargaining over trading profits; the output
is still a reported forecast. \citet{honguyen2026} retain a multivariate
predictive distribution and select its reconciliation map under pinball loss;
Theorem~\ref{thm:commute} explains why the marginal quantiles of the
resulting coherent distribution need not add, and why that distributional input
cannot be replaced by reconciled marginal summaries.

Across the distribution, report and action layers, the preceding results
identify three distinct sources of disagreement: the observation rule needed
to define a constrained distribution, summary extraction that does not commute
with reconciliation, and a decision loss that uses distributional information
beyond the mean. The classical benchmark studied next hides these distinctions
at the point-summary level: distributional revision, mean reporting and matched
quadratic action all select the same point even though the underlying
forecasting objects remain distinct.

\section{Why the classical benchmark hides the distinctions}
\label{sec:gauss}

Here the classical benchmark combines affine constraint geometry, Gaussian
predictive uncertainty and a quadratic loss whose curvature is matched to
forecast covariance. Under
this alignment, conditional revision, weighted projection, mean reporting and
constrained decision all select $T(\bmu)$. A comparison restricted to this
common point cannot reveal which operation produced it and therefore hides the
preceding distinctions. Proposition~\ref{prop:gauss} records the precise
agreement. This agreement is point-specific: distinct distributional
revisions can share the same mean, and nonmedian marginal quantiles remain
noncommutative.

\begin{proposition}[Equivalence under affine--Gaussian--quadratic assumptions]\label{prop:gauss}
Let $\bY\sim N_n(\bmu,\bSigma)$ with $\bSigma$ positive definite,
$\bA$ of full row rank and $\cC=\{\by:\bA\by=\bb\}$. Define
\[
T(\by)=\by-\bSigma\bA\T(\bA\bSigma\bA\T)^{-1}(\bA\by-\bb).
\]
Then
\begin{enumerate}[label=(\roman*),leftmargin=2.2em,itemsep=1pt]
\item $T_{\pf}P$ equals the Gaussian conditional distribution
$P(\cdot\mid \bA\bY=\bb)$;
\item the $\bSigma$-weighted point projection, the conditional mean and the
mean of $T(\bY)$ all equal $T(\bmu)$;
\item under
$L(\by,\ba)=(\by-\ba)\T\bSigma^{-1}(\by-\ba)$ the constrained Bayes action
equals $T(\bmu)$;
\item among $q\ll P$ with finite $\KL(q\Vert P)$ and finite first moment,
the $\KL$-minimising distribution subject to $\E_q(\bA\bY)=\bb$ is Gaussian with mean
$T(\bmu)$ and covariance $\bSigma$; it shares the point summary but not the
degenerate conditional distribution.
\end{enumerate}
\end{proposition}

The distributional identities in parts (i)--(ii) appear in Gaussian
probabilistic reconciliation
\citep{wickramasuriya2024,zambon2024efficient,zambon2024properties}. Part (iii)
combines the geometric weighted projection of \citet{panagiotelis2021} with
the constrained quadratic decision analysis of \citet{west2024}. Part (iv) is
the classical information-projection calculation of \citet{csiszar1975}.
Collected at a common point in the forecasting sequence, these identities
concern different transitions. Part (i) aligns conditional and transport-based
distributional revision once the Gaussian law and canonical affine residual
are specified; the constraint set alone still does not select that law. Part
(ii) removes the curvature displacement and recovers commutation for the mean
report. Part (iii) makes the projected mean a constrained Bayes action by
matching loss curvature to forecast covariance. Part (iv) marks the remaining
distributional boundary: an information projection can share the same mean
without sharing the conditional distribution's degenerate support.

The common point is what hides the distinctions in the classical benchmark. A
comparison based only on the mean or the matched quadratic action cannot
separate conditioning, projection, mean reporting and decision. Nonmedian
marginal quantiles lie outside this equivalence: the gap in
Section~\ref{sec:quantiles} remains at every nonmedian probability level, so
agreement of the point summaries does not make quantile additivity a Gaussian
identity.

\FloatBarrier

\section{Empirical evidence on quantile aggregation and target-dependent rankings}
\label{sec:empirical}

The empirical study starts from the affine--Gaussian part of the classical
benchmark and has two aims. The first exposes a report-level order effect:
it replaces the mean by a nonmedian marginal quantile and contrasts the two
paths in Figure~\ref{fig:square}. The second is a complementary evaluation
illustration: holding the predictive distribution and reconciliation map
fixed, it compares regional mean and $0.8$-quantile reports under target-specific
losses. No common capacity constraint is imposed, so this
comparison is not an empirical implementation of the constrained Bayes problem
in Section~\ref{sec:decision}. The rolling exercise therefore provides evidence
on quantile aggregation and an illustration of target-dependent evaluation,
without ranking forecasting methods by a single accuracy measure.

\subsection{Data and fixed forecast protocol}

We use the quarterly Australian domestic tourism data \texttt{tourism},
distributed with the R package \texttt{tsibble}. The series measure
thousands of overnight trips from 1998Q1 through 2017Q4
\citep{athanasopoulos2009,hyndmanathanasopoulos2021}. Tourism data are a
standard reconciliation benchmark \citep{athanasopoulos2024,bertani2025};
\citet{honguyen2026} also
study tourism under quantile-targeted reconciliation. Here the data serve a
different role: we hold the predictive law and reconciliation map fixed to
isolate the order of operations. Aggregating over travel purpose and
nesting Region within State gives the geographical
hierarchy used in the book example: $n=85$ series comprising
one national total, eight states and $m=76$ bottom-level regions.

We evaluate 28 one-quarter-ahead forecasts using an expanding window. The
forecast origins---the final quarters observed before the forecasts are
issued---run from 2010Q4 through 2017Q3, so the realised targets run from
2011Q1 through 2017Q4. At each origin an ETS model is fitted separately to
all 85 series
\citep{hyndmanathanasopoulos2021}.
Let $\widehat{\bmu}_o$ be the vector of base means and let
$\widehat{\bW}_o$ be the shrinkage estimate of the one-step base-error
covariance computed from information available at origin $o$. The estimate
uses response-scale one-step ETS residuals. This common scale is essential
because automatic ETS may select either additive- or multiplicative-error
models. The Schäfer--Strimmer estimator used by MinT shrinks the residual
correlation matrix toward the identity
\citep{schaeferstrimmer2005,wickramasuriya2019}. The resulting Gaussian working distribution is
$P_o=N_{85}(\widehat{\bmu}_o,\widehat{\bW}_o)$. Bottom-up and MinT-shrink
are applied as affine maps to the same $P_o$. At each origin, model fitting
and covariance estimation use only information available at that origin; no
specification is tuned using future outcomes or aggregate evaluation results.

Figure~\ref{fig:forecastpaths} makes the rolling forecast protocol visible at
three hierarchical levels. The displayed state and region are selected by
average trip volume in the initial training window ending in 2010Q4, before
any evaluation outcome is observed. This rule yields New South Wales and
Sydney. The figure is illustrative: all comparisons below use every one of
the 76 regions and all 28 forecast origins.

\begin{figure}[!ht]
\centering
\includegraphics[width=0.96\textwidth]{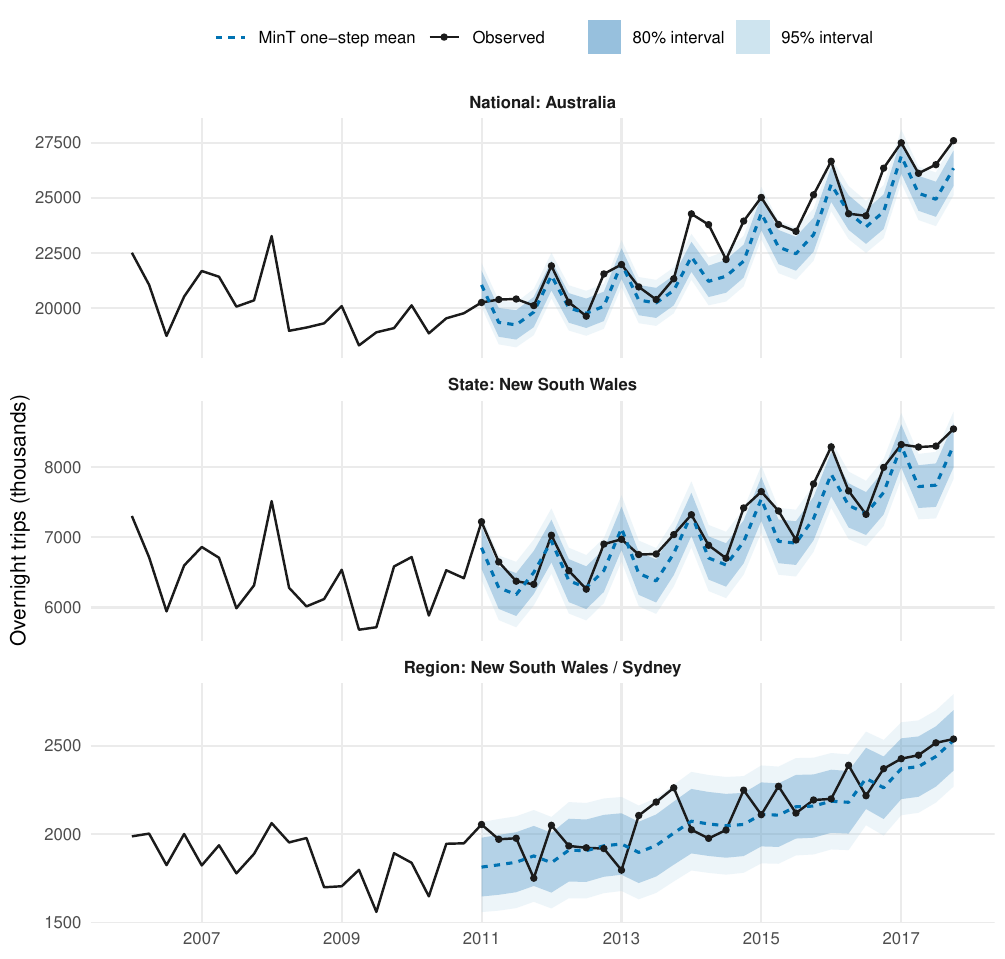}
\caption{Observed tourism sequences and rolling one-quarter-ahead
MinT-shrink forecasts. Solid lines with points are realised trips, dashed lines
are reconciled predictive means, and the shaded bands are pointwise Gaussian
$80\%$ and $95\%$ intervals. The pre-evaluation history starts in 2006Q1;
forecasts and intervals run from 2011Q1 through 2017Q4. Australia is the
national series, while the state and region panels are the largest by average
volume in the initial training window.}
\label{fig:forecastpaths}
\end{figure}

\subsection{Quantile aggregation: mechanism and magnitude}

The source of the quantile gap is visible before any loss is evaluated. Let
$\bM_o^{\mathrm{BU}}$ denote the bottom-up map at origin $o$. At
$\alpha=0.95$, Figure~\ref{fig:e2paths} compares the two orders,
\begin{equation}\label{eq:empiricalpaths}
\bM_o^{\mathrm{BU}}q_{0.95}(P_o)
\qquad\text{and}\qquad
q_{0.95}\bigl((\bM_o^{\mathrm{BU}})_\#P_o\bigr).
\end{equation}
For a nonnegative bottom-up row, Proposition~\ref{prop:gap}(ii) predicts that
the first route lies weakly above the second, with equality for a coordinate
row. The national and New South Wales rows aggregate multiple regional
coordinates: their gaps are respectively $3184.9$--$3275.7$ and
$630.7$--$651.6$ thousand trips over the 28 origins. Bottom-up leaves the
Sydney coordinate unchanged, and its two paths coincide exactly. Thus the
three panels display the aggregation mechanism that creates the quantile gap,
rather than only its average magnitude.

\begin{figure}[!ht]
\centering
\includegraphics[width=0.96\textwidth]{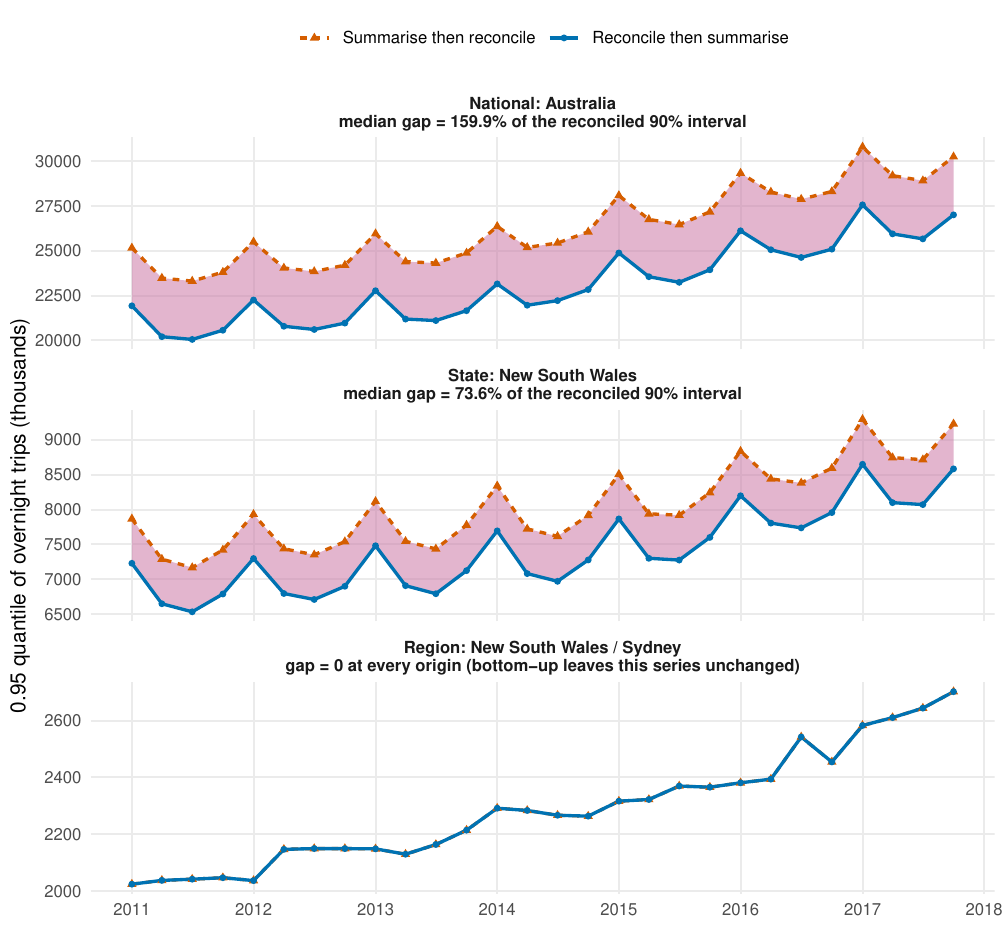}
\caption{Two ways to obtain bottom-up $0.95$ quantiles. The dashed
line first takes the base marginal quantiles and then reconciles them,
$\bM_o^{\mathrm{BU}}q_{0.95}(P_o)$; the solid line first reconciles the
predictive distribution and then takes its marginal quantiles,
$q_{0.95}((\bM_o^{\mathrm{BU}})_\#P_o)$. Shading shows the difference. The
aggregate series differ at every origin, whereas the two calculations agree
for the unchanged bottom-level series. The displayed nodes follow the
initial-training-volume rule used in Figure~\ref{fig:forecastpaths}.}
\label{fig:e2paths}
\end{figure}

At every origin and for
$\alpha\in\{0.05,0.50,0.95\}$, we also compare these two routes at all 85
nodes. The absolute difference is divided by the node's reconciled $90\%$
interval width and summarised by hierarchical level. Bottom-up has
nonnegative rows, so Equation~\eqref{eq:gap} gives the diversification gap
in Proposition~\ref{prop:gap}(ii). MinT has signed rows at every state and
regional node in this application; its commutator remains well defined, but
does not have the same sign interpretation.

Gaussian symmetry makes the absolute gaps at $\alpha=0.05$ and $0.95$
identical, while the median gap is zero. At either nonmedian tail, the
bottom-up median national gap is $159.88\%$ of the national interval width,
and the state median is $64.34\%$; the regional gap is zero because bottom-up
leaves the bottom series unchanged. For MinT-shrink, the corresponding
national, state and regional medians are $115.32\%$, $38.45\%$ and
$10.92\%$. An additive table of marginal quantiles can therefore differ from
the marginal quantiles of the coherent distribution by more than that
distribution's own predictive interval at the national level. The pathwise
plot and the all-node summary establish both the mechanism and its material
magnitude in this hierarchy.

\FloatBarrier

\subsection{Target-dependent evaluation of reported outputs}

Different losses elicit different functionals \citep{gneiting2011,kolassa2023},
and hierarchical evaluation can involve multiple decision objectives
\citep{athanasopouloskourentzes2023}. Holding the estimated predictive
distribution and reconciliation map fixed isolates the effect of the declared
target on the comparison. The exercise compares outputs constructed for
different reporting and evaluation tasks.

Let $\bM_o$ be the MinT-shrink map, and let subscript $b$ select the 76
regional coordinates. The same reconciled Gaussian distribution generates two
outputs:
\begin{equation}\label{eq:e1outputs}
\ba_o^{\mathrm{mean}}=(\bM_o\widehat{\bmu}_o)_b,
\qquad
\ba_o^{0.8}=\ba_o^{\mathrm{mean}}+z_{0.8}
\sqrt{\diag\{(\bM_o\widehat{\bW}_o\bM_o\T)_{bb}\}}.
\end{equation}
The first is the regional mean report. The second is the regional $0.8$-quantile
report, which coincides coordinatewise with the unconstrained regional Bayes
action for the loss with unit overage cost and underage cost four. Extending
either vector to all levels as $\bS\ba_o$ produces a coherent 85-vector.
For $\ba_o^{0.8}$, the upper-level entries are sums of regional $0.8$-quantile
reports rather than the corresponding upper-level marginal quantiles; their
distinction is another instance of the quantile aggregation gap examined above.

Let $d_{io}$ be region $i$'s seasonal-naive scale computed from the training
sample. For realised regional demand $\by_o$, the two declared evaluation
targets are
\begin{align}
S_o(\ba)&=\frac1{76}\sum_{i=1}^{76}
  \left(\frac{a_i-y_{io}}{d_{io}}\right)^2,\label{eq:e1square}\\
D_o(\ba)&=\frac1{76}\sum_{i=1}^{76}
  \frac{(a_i-y_{io})_++4(y_{io}-a_i)_+}{d_{io}}.\label{eq:e1decision}
\end{align}
The scaling prevents large regions from dominating either criterion. The
probability level $0.8$ is fixed by the $4{:}1$ ratio rather than selected
from the evaluation sample. Under the declared Gaussian distribution,
$\ba_o^{\mathrm{mean}}$ minimises expected squared error in
\eqref{eq:e1square}, whereas $\ba_o^{0.8}$ minimises the expected $4{:}1$
loss in \eqref{eq:e1decision}. A paired circular-block bootstrap with
four-quarter blocks describes variation across forecast origins
\citep{politisromano1992}.

\begin{table}[!ht]
\centering
\footnotesize
\caption{Target-specific evaluation over 28 rolling one-quarter-ahead forecasts.
Lower values are
better. Brackets contain the $95\%$ paired circular-block bootstrap interval
for the sample-mean action-minus-mean difference; the interval is
descriptive rather than a decision rule.}
\label{tab:e1}
\begin{tabular}{lrrr}
\toprule
Target & Mean report & $0.8$ action & Action $-$ mean \\
\midrule
Scaled squared error & 1.3118 & 1.4906 & 0.1788 $[-0.0016,\,0.3476]$ \\
Scaled $4{:}1$ loss & 2.6148 & 1.6679 & $-0.9468$ $[-1.1048,\,-0.7749]$ \\
\bottomrule
\end{tabular}
\end{table}

The sample-average ordering reverses. Figure~\ref{fig:e1paths} shows that the
mean has lower squared error at 19 of 28 origins, whereas the $0.8$ action
has lower $4{:}1$ loss at all 28. The squared-error interval in
Table~\ref{tab:e1} slightly overlaps zero, so the exercise does not provide
clear evidence that the average action-minus-mean loss difference is positive.
The ranking reversal is therefore descriptive.
Because the predictive distribution, reconciliation map and realised data
are identical across columns, changing the declared target changes the
comparison without changing forecast estimation. The columns are therefore
outputs for different tasks rather than competing estimates of a single
target.

\begin{figure}[!ht]
\centering
\includegraphics[width=\textwidth]{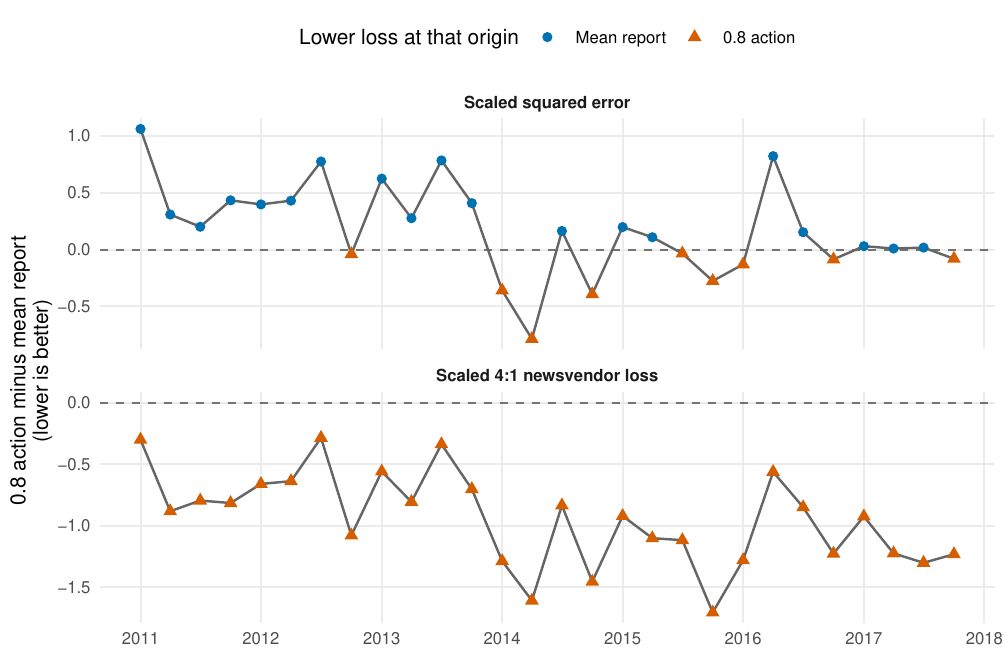}
\caption{Action-minus-mean loss differences at each forecast origin.
Positive values favour the coherent mean report and negative values favour
the coherent $0.8$ action; point symbols identify the lower-loss output at
that origin. The upper panel allows origin-specific exceptions but has a
positive mean, while every difference in the lower panel is negative.}
\label{fig:e1paths}
\end{figure}

Together, the two comparisons establish distinct empirical messages. Only the
quantile comparison is an order-effect experiment; the target comparison
illustrates a target-dependent ranking between reported outputs constructed for
different tasks.

\FloatBarrier

\section{Discussion and conclusion}\label{sec:conclusion}

Coherence answers whether an output respects an aggregation relation. It does
not determine where the relation originates, which output must satisfy it or
how a coherent candidate is selected. The affine--Gaussian--quadratic benchmark
hides these distinctions at the point-summary level: several standard procedures
select the same point and therefore appear interchangeable. Outside that
setting---and for nonmedian marginal quantiles even within it---the stage at
which a constraint is imposed, the summary reported and the decision target can
change the answer.

Proposition~\ref{prop:tube} treats smooth regular level sets and one
family of shrinking-tolerance limits. Theorem~\ref{thm:commute} treats
deterministic affine maps and marginal quantiles; corresponding commutation
results for learned or nonlinear maps, expectiles and path summaries remain
open. Proposition~\ref{prop:como} concerns every probability level; finite
reporting grids are weaker. The formal results are population-level, while
soft, noisy or learned aggregation relations require a model for constraint
uncertainty.

The tourism quantile experiment establishes existence, material magnitude and
reproducibility of the report-level gap in one hierarchy. The target-specific
comparison illustrates one ranking reversal between reported outputs. Neither
part estimates the prevalence of quantile aggregation gaps or target-dependent
rankings across data sets. The base-forecast and $\bW$ estimation protocols are
held fixed, and sensitivity to estimation error remains outside the experiment.

These results also delimit valid method comparisons. Existing guidance on
hierarchical evaluation emphasises scale, repeated evaluation windows and
multiple decision objectives \citep{athanasopouloskourentzes2023}. The
comparison contract in \eqref{eq:comparison-contract} adds the identity of the
constraint source and constrained output. When
operational risk is the target, the induced distribution of $L(\bY,\ba)$ may
matter beyond its expectation \citep{west2024}; this changes how a specified
output is judged without changing the reconciliation specification.

Two procedures are
directly comparable only when their constraint source, constrained output and
evaluation target coincide; the procedures can then be evaluated as alternative
solutions to the same task. If the source or output differs, a
common downstream score measures the consequence of substituting one output
for another rather than ranking solutions to a single reconciliation problem.
An empirical report should therefore state these three elements and the
selection criterion, together with the representation used for a nonlinear
constraint or the probability levels at which quantile coherence is imposed.

Aggregation coherence becomes a meaningful comparison criterion only relative
to a specified output; predictive distributions, reported forecasts and
constrained actions remain distinct forecasting objects.

\section*{Acknowledgements}
We thank Prof. Rob J. Hyndman for his helpful comments and suggestions on an earlier version of this paper.

\begin{samepage}
\section*{Data availability statement}
The Australian tourism data analysed in this study are publicly available as
the \texttt{tourism} data set in the R package \texttt{tsibble}. R code for
reproducing the analytical examples and tourism analysis, together with derived
results, figures and package-version information, is available at
\url{https://github.com/Eurekacoding/coherence-is-not-enough}.

\end{samepage}

\section*{Conflict of interest}
The authors declare no conflicts of interest.

\begin{Appendix}\label{app:proofs}
\subsection*{Proofs}

\subsection*{Proof of Proposition~\ref{prop:tube}}
For part (i), the co-area formula gives
\[
\int_{\{\lVert h(\by)\rVert\leq\epsilon\}}f(\by)p_0(\by)\,d\by
=\int_{\lVert z\rVert\leq\epsilon}
  \int_{h^{-1}(z)}\frac{f(\by)p_0(\by)}{J_h(\by)}
  \,d\Haus^{n-k}(\by)\,dz .
\]
The continuity and domination assumptions make the inner integral equal
to its value at zero plus $o(1)$ after averaging over the shrinking ball.
Applying the same argument with $f=1$ and taking the ratio yields
\eqref{eq:tube}; see also \citet{diaconis2013}. For part (ii), let $r$ be
the target density ratio and extend $r^{-1/k}$ to a positive smooth function
$\phi$ on a tubular neighbourhood and then to $\R^n$. The function
$\widetilde h=\phi h$ has the same zero set, and
$D\widetilde h=\phi\,Dh$ on $\cM$, so
$J_{\widetilde h}=\phi^{k}J_h$ and the density on $\cM$ is multiplied by
$r$. Compactness gives normalisability. More generally,
$h\mapsto\Gfield h$ with $\Gfield$ smooth and invertible on $\cM$ produces
the factor $|\det\Gfield|^{-1}$. \hfill$\square$

\subsection*{Proof of Proposition~\ref{prop:curv}}
Taylor expansion gives
$g(\bu)=\tfrac12\bu\T\mathrm{II}_{\by_0}\bu+o(\lVert\bu\rVert^2)$;
taking expectations under the assumed domination,
$\E g(\bU^\epsilon)=\tfrac12\tr[\mathrm{II}_{\by_0}
\E(\bU^\epsilon\bU^{\epsilon\top})]
+o[\E\lVert\bU^\epsilon\rVert^2]$. Tangential centring makes the second
moment $\bSigma_{T,\epsilon}$. The tangential component of
$\E(\bY^\epsilon)-\by_0$ vanishes, so projecting onto $\bnu$ gives
\eqref{eq:curv}; diagonalising $\mathrm{II}_{\by_0}$ gives the principal form.
\hfill$\square$

\subsection*{Proof of Theorem~\ref{thm:commute}}
Let $z_\alpha$ be the standard normal $\alpha$-quantile, nonzero by
assumption. For a Gaussian distribution with mean $\bmu$ and covariance
$\diag(\sigma_1^2,\dots,\sigma_n^2)$, the $i$th coordinates of the two
sides of \eqref{eq:commute} are
\[
(\bM\bmu+\bd)_i+z_\alpha\Bigl(\sum_j \bM_{ij}^2\sigma_j^2\Bigr)^{1/2}
\quad\text{and}\quad
(\bM\bmu+\bd)_i+z_\alpha\sum_j \bM_{ij}\sigma_j .
\]
Equality for every positive $(\sigma_1,\dots,\sigma_n)\T$ gives
$\bigl(\sum_j\bM_{ij}^2\sigma_j^2\bigr)^{1/2}
=\sum_j\bM_{ij}\sigma_j$.
Squaring, the polynomial
$2\sum_{j<k}\bM_{ij}\bM_{ik}\sigma_j\sigma_k$ vanishes
on the positive orthant, so every coefficient is zero and at most one
entry of the row is nonzero; the unsquared equality makes it nonnegative.
Conversely a zero row returns the constant $d_i$, and a row $c_ie_j\T$
with $c_i>0$ uses $Q_{c_iY_j+d_i}(\alpha)=c_iQ_{Y_j}(\alpha)+d_i$, valid
for arbitrary marginals. If $\bM^2=\bM$ and row $i$ is $c_ie_{\pi(i)}\T$, then
row $i$ of $\bM^2$ is $c_i$ times row $\pi(i)$ of $\bM$, forcing row $\pi(i)$
to equal $e_{\pi(i)}\T$. \hfill$\square$

\subsection*{Proof of Proposition~\ref{prop:gap}}
For (i), the marginal standard-deviation vector of $T(\bY)$ is
$\sqrt{\diag(\bM\bSigma\bM\T)}$, so the Gaussian quantile formula gives
\eqref{eq:gap} after the common location $\bM\bmu+\bd$ is cancelled. Under the
conditions of (ii), put $X_j=Y_j-\E Y_j$. Minkowski's inequality in
$L^2$ gives
\[
\operatorname{sd}\!\left(\sum_j\bM_{ij}Y_j\right)
=\left\lVert\sum_j\bM_{ij}X_j\right\rVert_2
\leq\sum_j\bM_{ij}\lVert X_j\rVert_2
=\sum_j\bM_{ij}\sigma_j.
\]
Equality in the triangle inequality holds exactly when the nonzero
$\bM_{ij}X_j$ are nonnegative scalar multiples of a common element of
$L^2$.
For a Gaussian vector this is equivalent to perfect positive correlation
among the coordinates receiving positive weight; positive definiteness of
$\bSigma$ therefore leaves at most one such coordinate. For (iii), a comonotonic vector satisfies
$\bY=\bigl(Q_{Y_1}(U),\dots,Q_{Y_n}(U)\bigr)$ for a common uniform $U$.
Each output coordinate $\sum_j\bM_{ij}Y_j+d_i$ is a nondecreasing function
of $U$, whose quantile is obtained by evaluation at $\alpha$. This yields
$\sum_j\bM_{ij}Q_{Y_j}(\alpha)+d_i$. \hfill$\square$

\subsection*{Proof of Corollary~\ref{thm:bregman}}
The integrability assumptions justify taking expectations after expanding
$D_\varphi+h$, which gives the displayed risk. Its action-dependent part has gradient
$-\nabla^2\varphi(\ba)(\bmu-\ba)$;
because $\mathcal D$ is open and $\bA$ has full row rank, the Lagrange multiplier
condition at any constrained minimiser yields
$\nabla^2\varphi(\ba^{\star})(\bmu-\ba^{\star})
=\bA\T\blambda$ with $\bA\ba^{\star}=\bb$. For
$\varphi=\tfrac12\by\T\bW^{-1}\by$ this is
$\ba=\bmu-\bW\bA\T\blambda$, and $\bA\ba=\bb$
gives $\blambda=(\bA\bW\bA\T)^{-1}(\bA\bmu-\bb)$, so
$\ba=T_{\bW}(\bmu)$. Strict convexity of the resulting quadratic risk on
the affine feasible set gives uniqueness.
\hfill$\square$

\subsection*{Proof of Proposition~\ref{prop:nogo}}
A mean-based rule takes the common value $\Psi(\bmu)$ under both counterexamples
while the unique actions differ, so it fails for at least one.
For the two gamma witnesses reported after the proposition, substituting the
corresponding distribution functions into the first-order equation
$R_1'(a_1)=R_2'(85-a_1)$ gives the two reported actions; strict
monotonicity of the distribution functions gives uniqueness.
\hfill$\square$

\subsection*{Proof of Proposition~\ref{prop:gauss}}
Set $\bK=\bA\bSigma\bA\T$. The Gaussian conditional has mean
$\bmu+\bSigma\bA\T\bK^{-1}(\bb-\bA\bmu)=T(\bmu)$ and covariance
$\bSigma-\bSigma\bA\T\bK^{-1}\bA\bSigma$; the affine image
$T(\bY)$ is Gaussian with
exactly these moments, proving (i) and (ii). Expected quadratic loss
differs from $(\ba-\bmu)\T\bSigma^{-1}(\ba-\bmu)$ by a constant, whose constrained
minimiser is $T(\bmu)$, proving (iii); this is
Corollary~\ref{thm:bregman} with $\bW=\bSigma$. For (iv), multiplying the
density of $P$ by $\exp(\blambda\T\bA\by)$ gives a Gaussian distribution
with covariance $\bSigma$ and mean
$\bmu+\bSigma\bA\T\blambda$; the constraint sets
$\blambda=\bK^{-1}(\bb-\bA\bmu)$, and
optimality follows from standard $\KL$-projection theory \citep{csiszar1975}.
\hfill$\square$

\end{Appendix}

\begin{Appendix}\label{app:gamma}
\subsection*{Two-store gamma illustration}

For $Y\sim\mathrm{Gamma}(k,\theta)$ with $\mu=k\theta$, and $F_k,F_{k+1}$
gamma distribution functions with common scale $\theta$,
\[
\begin{aligned}
\E\bigl((a-Y)_{+}\bigr)&=aF_k(a)-\mu F_{k+1}(a),\\
\E\bigl((Y-a)_{+}\bigr)&=\mu\{1-F_{k+1}(a)\}
  -a\{1-F_k(a)\},
\end{aligned}
\]
which give the inventory-loss column of Table~\ref{tab:six} without simulation.

\emph{Information update.} The density of $Y_1$ given $Y_1+Y_2=\tgt$ is
proportional to $f_1(x)f_2(\tgt-x)$ on $(0,\tgt)$; quadrature gives mean
$27.5146$. The normalising density is
$f_{Y_1+Y_2}(85)=0.0179377$, giving the total log score
$-\log f_{Y_1+Y_2}(85)=4.02085$ reported in Table~\ref{tab:six}.

\emph{Distribution revision.} With rates $\beta_i=1/\theta_i$, multiplying
the density by $\exp\{\lambda(Y_1+Y_2)\}$ changes the rates to
$\beta_i-\lambda$, and the
unique root of $\sum_i k_i/(\beta_i-\lambda)=\tgt$ is
$\lambda=-0.03939$, giving $(28.6964,56.3036)\T$. With
$\psi(\lambda)=\sum_i k_i\log\{\beta_i/(\beta_i-\lambda)\}$, the attained
divergence is $\lambda\tgt-\psi(\lambda)=0.26933$.

\emph{Dependence bounds.} For fixed marginals, the sharp lower and upper
distribution-function bounds for the sum are
\[
\underline F_{+}(s)=\sup_x\{F_1(x)+F_2(s-x)-1\}_{+},\qquad
\overline F_{+}(s)=\inf_x\{F_1(x)+F_2(s-x)\}\wedge1.
\]
Numerical inversion at probability level $0.95$ gives the interval
$[114.0128,166.4133]$ reported in Section~\ref{sec:como}
\citep{embrechts2013}.

\emph{Reporting.} The weighted projection is
$\bmu-\bW\bone(\bone\T\bW\bone)^{-1}(\bone\T\bmu-\tgt)$ with
$\bone$ the two-vector of ones; $\bW=I$ and
$\bW=\diag(400,100)$ give the
two reporting rows.

\emph{Decision.} With $R_i'(a)=(c^o_i+c^u_i)F_i(a)-c^u_i$, substituting
$a_2=\tgt-a_1$ and solving $R_1'(a_1)=R_2'(a_2)$ gives
$(51.1140,33.8860)\T$. Reversing the shapes to $(36,4)$ while holding the
means at $(40,60)\T$ gives $(45.9117,39.0883)\T$, the equal-mean counterexample in
Proposition~\ref{prop:nogo}.

\end{Appendix}

\begin{Appendix}\label{app:revenue}
\subsection*{Lognormal revenue illustration}

Let $Y_1$ be price and $Y_2$ volume, both positive, and consider
conditioning on an observed revenue $Y_1Y_2=c$. Two natural functions for expressing the constraint are
$h_1(y)=y_1y_2-c$ and $h_2(y)=y_1-c/y_2$. They define the same curve, but
$\lVert\nabla h_1\rVert=(y_1^2+y_2^2)^{1/2}$ and
$\lVert\nabla h_2\rVert=(y_1^2+y_2^2)^{1/2}/y_2$, so
Proposition~\ref{prop:tube} gives
\begin{equation}\label{eq:tilt}
p_0^{h_2}(y)\ \propto\ y_2\,p_0^{h_1}(y),\qquad y_1y_2=c .
\end{equation}
The second protocol assigns more mass to high-volume, low-price points on
the same curve.

Let $X_i=\log Y_i$, $i=1,2$, be independent, with mean vector
$(\log5,\log20)\T$ and variance vector $(1,1.44)\T$, and set $c=100$.
Parameterising the curve by
$x_2=\log y_2$, the surface element and co-area Jacobian give an $X_2$
density proportional to $f_{X_1}(\log c-x_2)f_{X_2}(x_2)$ under $h_1$.
This product of normal densities yields $X_2\sim N(\log20,v)$, where
$v=\sigma_1^2\sigma_2^2/(\sigma_1^2+\sigma_2^2)=1.44/2.44=0.5902$.
Multiplication by $y_2=e^{X_2}$ in \eqref{eq:tilt} shifts the normal mean,
so under $h_2$ it becomes
$\log20+v$ with variance unchanged. Hence
\[
\E_{p_0^{h_1}}(Y)=(6.72,\,26.86)\T,\qquad
\E_{p_0^{h_2}}(Y)=(3.72,\,48.47)\T,
\]
and every volume quantile under the second protocol is
$e^{v}=1.804$ times its counterpart under the first. Neither protocol is
selected by the revenue curve.

Both mean pairs have product $100e^{v}=180.43$. Indeed,
$\E Y_1\,\E Y_2=c\,e^{v}$ under this exponential reweighting of the
conditional distribution of $X_2$, because $Y_1=ce^{-X_2}$ and the
reweighting shifts the mean of $X_2$ without changing $v$. The mean
therefore leaves the constraint set
by the same factor under both protocols, illustrating the
curvature effect quantified in Proposition~\ref{prop:curv}. Here
the departure is $80\%$ of $c$, far outside the small-variance regime
in which \eqref{eq:curv} is a good approximation.

For a small-variance check, scale both log variances by $t\downarrow0$ and
set $\by_0=(5,20)\T$, $\bnu=(20,5)\T/\sqrt{425}$. The exact displacement
in the normal direction is
\[
\langle \E(\bY^{t})-\by_0,\bnu\rangle
=\frac{200}{\sqrt{425}}\{e^{tv/2}-1\}
=\frac{100v}{\sqrt{425}}t+O(t^2).
\]
Here the signed curvature is $200/(425\sqrt{425})$ and the leading
tangential variance is $425vt$, so \eqref{eq:curv} gives the same leading
term.

\end{Appendix}

\bibliographystyle{plainnat}
\bibliography{references}

\end{document}